\documentclass[12pt]{article}

\usepackage{a4}

\usepackage{latexsym,amsfonts,amssymb,amsmath,amsthm}

\usepackage[T1]{fontenc}
\usepackage{color}

\def\div{\mathop{\rm div}}

\def\const{\mathop{\rm const}}

\def\tr{\mathop{\rm tr}}
\def\t0{\mathop{\overset{0}{\rm tr}}}
\def\lin{\mathop{\rm span}}

\def\DO{\mathop{{\rm d}o}}

\newtheorem{thm}{Theorem}[section]

\newtheorem{prop}[thm]{Proposition}

\theoremstyle{definition}
\newtheorem{defn}[thm]{Definition}
\newtheorem{rem}[thm]{Remark}

\newcommand{\thmref}[1]{Theorem~\ref{#1}}

\newcommand{\propref}[1]{Proposition~\ref{#1}}
\newcommand{\defnref}[1]{Definition~\ref{#1}}
\newcommand{\remref}[1]{Remark~\ref{#1}}

\begin{document}

\title{On the existence of rigid spheres in four-dimensional spacetime manifolds}

\author{{\em Hans-Peter Gittel}
\thanks{E-mail: gittel@math.uni-leipzig.de} \\
Department of Mathematics, University of Leipzig, \\
Augustusplatz~10, 04109~Leipzig, Germany
\and {\em Jacek Jezierski}
\thanks{E-mail: Jacek.Jezierski@fuw.edu.pl} \\
 Department of Mathematical Methods in
	Physics, \\ University of Warsaw,
	ul. Pasteura 5, 02-093 Warszawa, Poland
\and  {\em Jerzy Kijowski}
\thanks{E-mail: kijowski@cft.edu.pl} \\
Centrum Fizyki Teoretycznej PAN, \\
Aleja Lotnik\'ow 32/46, 02-668 Warsaw, Poland }


\maketitle

\begin{center}
\Large\emph{Dedicated to Professor Dr.~Eberhard Zeidler\\
on the occasion of his 75th birthday}
\end{center}


\newpage
\begin{abstract}
This paper deals with the generalization of usual round spheres in the flat Minkowski spacetime to the case of a generic four-dimensional spacetime manifold $M$. We consider geometric properties of sphere-like submanifolds in $M$ and impose conditions on its extrinsic geometry, which lead to a definition of a {\em rigid sphere}. The main result is a local existence theorem concernig such spheres. For this purpose we apply the surjective implicit function theorem. The proof is based on a detailed analysis of the linearized problem and leads to an
eight-parameter family of solutions in case when the metric tensor $g$ of $M$ is from a certain neighbourhood of the flat Minkowski metric. This contribution continues the study of rigid spheres in~\cite{GJKL}.
\end{abstract}

\numberwithin{equation}{section}

\section{Introduction}

In \cite{GJKL} we have introduced the concept of ``rigidity'' of spheres in generic Riemannian three-manifolds and gave a short outline for four-dimensional Lorentzian spacetimes (cf.~\cite[Subsection 2.5]{GJKL}). This paper presents geometrical considerations which lead to the definition of rigid
spheres in this case and proves a local existence result concerning such spheres.

Let $M_0$ be the flat Minkowski spacetime, i.e.~the space
${\mathbb R}^4$ parameterized by the Lorentzian coordinates
$(x^\alpha)=(x^0, \ldots ,x^3)$ and equipped with the metric $\eta
= (\eta_{\alpha \beta}) = {\rm diag} (-1,1,1,1)$ (greek indices
run always from 0 to 3). \\
Consider in $M_0$ a {\em round sphere}, i.e.~the
two-dimensional submanifold defined by
\[
  S_{T,R} := \Big\lbrace x \in {\mathbb R}^4 \
  \Bigl| \ x^0 = T \ , \ \sum_{i=1}^3 (x^i)^2 = R^2 \Big\rbrace
   \ ,
\]
where the time $T\in {\mathbb R}$ and the sphere's radius $R>0$ are fixed.
It may be easily verified (cf.~Subsection 3.1) that the
submanifold fulfills the following conditions:
\[
\|{\bf k}\| = \tfrac {2}{R} =\left\langle  \|{\bf k}\|  \right\rangle = \const \, ,
\]
where 
$\|{\bf k}\| = \sqrt{\eta_{\alpha \beta} k^\alpha k^\beta}$, and
${\bf k} = (k^\alpha )$ denotes the extrinsic curvature vector of
$S_{T,R}$ which is orthogonal to $S_{T,R}$ (cf.
Definition \ref{Def1}).
Throughout the paper, $ \left\langle  f  \right\rangle$ is the abbreviation of the mean
value of a function $f$ over its domain (here: $S_{T,R}$).
Moreover, the {\em extrinsic torsion} (cf. Definition \ref{Def1}) of
$S_{T,R}$ vanishes:
\[
{\bf t} = 0 \ .
\]
Because of similarity in $M_0$, there is an
eight-parameter family of such spheres (cf. \cite[Subsection 2.5]{GJKL}).


Two-dimensional round spheres in $M_0$ may be used as building
blocks to generate important geometric objects like three-hyperplanes,
light cones, hyperboloids etc. These objects are necessary in the
hamiltonian description of the field evolution in the special
relativity,
cf.~\cite{CJK}.

In General Relativity Theory, the choice of a Cauchy hypersurface
plays the role of a specific gauge condition imposed on the
hamiltonian description of the field evolution. An intelligent use
of gauge may simplify considerably the analysis of a physical
problem (like e.g.~{\em maximal gauge} in the proof of the
positivity of gravitational energy or the Bondi condition in the
analysis of the radiation on the {\em Scri} hypersurface).


A long experience supports the conjecture that the ``good''
hypersurface is often a collection of ``good'' two-spheres (cf.~\cite{pieszy, LS} in connection with description of gravitational energy).
This problem is important especially
for an asymptotically flat spacetime, where the use of
such intrinsically defined surfaces could help us to construct a
``privileged'' parametrization, at least far away from
gravitational sources. It would serve as a tool of a deep analysis
of the structure of gravitational radiation at infinity.

Therefore a question arises, whether or not geometric conditions
may be used in case of a generic spacetime $M$, equipped
with a non-flat metric structure $g$, to select a family of
privileged two-surfaces.

Unfortunately, the conditions: $\|{\bf k}\| -  \left\langle  \|{\bf k}\|  \right\rangle = 0$ and
${\bf t} = 0$ form, in general, an over-determined system of
equations which may admit no solution. The existence of an
eight-parameter family of two-surfaces fulfilling these equations
in the Minkowski spacetime is not a generic phenomenon but is a
mere coincidence due to the maximal symmetry of $M_0$. In the present paper, we show how to relax the above conditions in such a way that: \\
1) they admit generically an eight-parameter family of solutions and,
moreover, \\
2) in  case of the Minkowski space, the solutions are
the same as before. \\
By analogy with the flat Minkowski geometry, a generic surface
satisfying our improved condition will be called a {\em rigid
sphere}.

The paper is organized as described below. In the next section, we recall some
geometric properties (extrinsic curvature and torsion, projections)
of sphere-like submanifolds in a generic four-dimensional spacetime and give
the definition of rigid
spheres. The existence of such spheres is a highly nonlinear problem
which is locally solved by application of an implicit function
theorem from functional analysis. For this purpose it is necessary
to analyze the linearized problem in detail. The linearization of the nonlinear problem is carried out in Section 3 whereas Section 4 contains the derivation of crucial properties of the linearized operator using Sobolev spaces $H^{k}$. Our main result is presented in the last section.


\section{Spheres in a generic spacetime}

\subsection{Topological spheres}

Let us consider now a generic four-dimensional spacetime manifold
$M$ with the metric tensor $g$ of signature $(-,+,+,+)$.
Any choice of local coordinates $(x^0, \ldots ,x^3)$ induces the
natural basis
\[
  {\bf e}_\alpha := \frac \partial{\partial x^\alpha},
  \qquad  \alpha = 0, \dots 3
\]
in fibers of the tangent bundle $T M$. If ${\bf v}=v^\alpha
{\bf e}_\alpha$ and ${\bf w}=w^\alpha {\bf e}_\alpha$,
then\footnote{Einstein summation convention over repeated lower
and upper indices is always assumed.} the inner product $g({\bf
v},{\bf w})$ on $M$ is given by $g({\bf v},{\bf w}) =
g_{\alpha
\beta} v^\alpha w^\beta$, where $g_{\alpha \beta}:= g({\bf
e}_\alpha , {\bf e}_\beta)$.

Moreover, let ${\cal S} \subset M$ be a {\em
  topological sphere}, i.e. a two-dimensional
submanifold diffeomorphic to the {\em standard unit sphere} $S^2
\subset {\mathbb R}^3$. The diffeomorphism in question  $F: S^2
\rightarrow M$, it is locally described by four functions
\begin{equation} \label{eq1a}
F: (u^2,u^3) \mapsto x=(x^{\alpha} ) \in M \ , \
x^{\alpha} = x^{\alpha} (u^2,u^3) \, ,
\end{equation}
where $(u^A)=(u^2,u^3)=(\theta, \phi)$ are standard spherical coordinates\footnote{Throughout the paper, capital letters $A,B, \ldots \! \in \! \{2,3
\},$ whereas small letters $a,b, \ldots  \!\in \! \{0,1 \}.$} on
$S^2$.
Using
parametrization (\ref{eq1a}) we may take vectors
\begin{equation}
\label{eq4a}
  {\bf b}_A = \partial_A x := \frac{\partial x}{\partial u^A } =
\frac{\partial x^\alpha}{\partial u^A } {\bf e}_\alpha
\end{equation}
as a basis in each tangent space $T_x {\cal S}$. Consequently, the
{\em induced metric} $s = F^*g$ on ${\cal S}$ is represented by
the metric tensor
\begin{equation} \label{eq4}
  s_{AB} = g ({\bf b}_A ,{\bf b}_B) =
  \frac {\partial x^{\alpha}}{\partial u^A} \
  g_{\alpha \beta} \
    \frac {\partial x^{\beta}}{\partial u^B}
   \ .
\end{equation}
Associated with metric $s=(s_{A B})$, the covariant derivative $
\overset{s}{\nabla}_A $ and the Laplace-Beltrami operator
$\overset{s}{\Delta}:= \overset{s}{\nabla}{^A}
\overset{s}{\nabla}_A = s^{A B} \overset{s}{\nabla}_{A}
\overset{s}{\nabla}_B$ on ${\cal S}$ are given. As usual, $(s^{A
B})$ is the inverse matrix of $s:=(s_{AB})$.

As an example of the settings above, we may take the system of
spherical coordinates in the Minkowski spacetime:
\begin{gather}
\label{eq1} x^0 = t \ , \ x^1 = r \sin \theta \cos \phi \ , \ x^2
= r \sin \theta \sin \phi \ , \ x^3 = r \cos \theta \, , \\ \label{eq2}
\text{where } t \in {\mathbb R}\ , \ r \in [0,\infty) \ , \ \theta
\in [ 0, \pi ] \ , \ \phi \in [0,2 \pi) \ ,
\end{gather}
and the round sphere
$ S_{T,R} = \{ t = T , r = R \}  .$
In these coordinates, Minkowski metric reads as
\begin{equation} \label{eq2a}
  \eta_{\alpha \beta} \mbox{\rm d}x^{\alpha} \mbox{\rm d}x^{\beta} = - (\mbox{\rm d}t)^2 +
  (\mbox{\rm d}r)^2 + r^2 \left( (\mbox{\rm d} \theta)^2 + \sin^2 \theta (\mbox{\rm d} \phi )^2 \right) \ .
\end{equation}
This formula implies that the induced metric $s_{AB}$ on
$S_{T,R}$ is proportional to the {\em standard round metric }
$\overset{\circ}{\eta}_{A B}$ on the unit sphere $S^2$:
\[
\label{int}
s_{AB} \mbox{\rm d}u^{A} \mbox{\rm d}u^{B} = R^2 \left( (\mbox{\rm d} \theta)^2 + \sin^2 \theta
(\mbox{\rm d} \phi)^2 \right) \ ,
\]
or, equivalently, $s_{AB} = R^2 \ \overset{\circ}{\eta}_{A B}$,
where
\begin{equation}
\label{eq1b} \left( \overset{\circ}{\eta}_{A B} \right) := \left(
\begin{matrix} 1 & 0 \\ 0 & \sin^2 \theta \end{matrix} \right) \ .
\end{equation}

\subsection{Extrinsic geometry of a topological sphere}

Let ${\cal S} \subset M$ be a topological sphere. Then the {\em extrinsic curvature tensor} ${\bf K}_{A B}$ of ${\cal S}$
is defined as the orthogonal (with respect to ${\cal S}$) part of
the covariant derivative of the tangent vector field ${\bf b}_B$
in direction of another tangent vector ${\bf b}_A$:
\begin{equation}
  \label{eq4c} {\bf K}_{A B} :=
  \left( \overset{g}{\nabla}_{{\bf b}_A}
  {\bf b}_B \right)^\perp \in T{\cal S}^\perp \ .
\end{equation}
Especially convenient expression for
the components of ${\bf K}_{A
B}$ is obtained if we choose coordinates $(x^\alpha)$ in a
neighbourhood of ${\cal S}$ in such a way that
the transversal
coordinates $x^0, x^1$ are constant on ${\cal S}$, consequently,
${\bf e}_2, {\bf e}_3$ are tangent to ${\cal S}$ and $(x^A)=(u^A)$.
Such coordinate system will be called {\em adapted to} ${\cal
S}$. Hence, ${\bf b}_A= {\bf e}_A$ and
\begin{equation} \label{eq6}
  \overset{g}{\nabla}_{{\bf b}_A}
  {\bf b}_B = G^{\alpha}_{A B} {\bf e}_{\alpha} =
   G^a_{A B} {\bf e}_a + G^C_{A B} {\bf e}_C \ ,
\end{equation}
with the Christoffel
symbols $G^{\alpha}_{\beta \gamma}$ of  metric $g.$
Using the projection ``$\perp$'' we have
\begin{equation}\label{b_a}
   {\bf b}_a := {\bf e}_a^\perp = {\bf e}_a - n_a^{\ C} {\bf e}_C \, , \quad
\text{where} \quad   n_a^{\ C} = g_{aB}\, s^{BC} \, ,
\end{equation}
as basis of the normal bundle $T{\cal S}^\perp$ and we obtain the
following formula for the components $K^a_{A B}$ of the vector
${\bf K}_{A B}=K^a_{A B} {\bf b}_a$, valid in adapted coordinates:
\begin{equation}\label{kAB}
    K^a_{A B} = G^a_{A B} \ .
\end{equation}

\begin{defn}
\label{Def1} The {\em extrinsic curvature vector} ${\bf k}$ of ${\cal S}$ is defined by
\begin{equation}\label{k}
    {\bf k} =
    s^{AB} {\bf K}_{A B} \in T{\cal S}^\perp \ ,
\end{equation}
Moreover, we consider its {\em extrinsic torsion}  at the point $x$:
\begin{equation} \label{eq6b}
t_A :=  g\big({\bf m} , \overset{g}{\nabla}_A {\bf n} \big) \ , \quad
{\bf t}(x)=t_A {\rm d}u^A \in T^*_x{\cal S}
\end{equation}
where
\begin{equation}\label{m}
    {\bf n} := \frac {\bf k}{\| {\bf k} \| } \qquad \text{({\em first} orthonormal vector)},
\end{equation}
${\| {\bf k} \| } = \sqrt{k^a g_{ab} k^b}$ provided ${\bf k}$ is
space-like ($\|{\bf k}\|^2 := g(${\bf k}$,${\bf k}$)>0$),
and ${\bf m}$ denotes the {\em second} orthonormal vector, i.e. the unit vector orthogonal to both ${\bf k}$ and $S$.
\end{defn}

\begin{rem}
\label{Rem-tersion}
Geometrically, $t_A$ represents the connection in the bundle of orthonormal vectors over ${\cal S}$. For any curve in ${\cal S}$ which is geodesic with respect to its intrinsic geometry $s_{AB}$, one might consider its ``Dreibein'' $({\bf v},{\bf n},{\bf m})$, where ${\bf v}$ is the tangent vector to the curve. Then
$g(\bf t,\bf v)$
describes the angular velocity of rotation of the ``Zweibein'' $({\bf n},{\bf m})$. By analogy with classical Frenet-Serret formulae, we call $\bf t$ the extrinsic torsion of ${\cal S}$, which shortens considerably our terminology.
\end{rem}

\begin{rem}
\label{Rem0}
Using adapted coordinates and setting ${\bf k} = k^a {\bf b}_a$,
\begin{equation}
\label{eq6a} k^a :
= s{^{A B}} G^a_{A B}  \ .
\end{equation}
follows. If ${\bf n} = n^a {\bf b}_a$, we get ${\bf m} = m^a {\bf b}_a$ with
\begin{equation}\label{n}
    m^a = s^{ab} {\varepsilon}_{bc} n^c \ .
\end{equation}
Here, $(s^{ab})$ is the two-dimensional inverse to the metric
$(s_{ab})$, which is given by
\begin{equation}\label{s}
    s_{ab} :=g({\bf b}_a,{\bf b}_b) =
    g_{ab} - n_a^{\ A} n_b^{\ B} g_{AB}
\end{equation}
(cf.~(\ref{b_a})),
and $\varepsilon_{ab} = \left| \det ( s_{cd} )
\right|^{1/2} E_{ab}$ indicates the {\em Levi-Civita tensor} for $(
s_{cd} )$ with $E_{ab}$ being the sign of the permutation $ \left(
\begin{smallmatrix} 0 & 1 \\ a
    & b \end{smallmatrix} \right).$
Of course,
it is always possible to choose an
extension $x^A$ of coordinates $u^A$ from ${\cal S}$ to its
neighbourhood in such a way that ${\bf e}_a$ are orthogonal to
${\cal S}$. In this case, $n_a^{\ A} = 0$, $s_{ab} =
g_{ab}$, and
the above formulae simplify substantially.
\end{rem}



Given a surface ${\cal S} \subset M$ and a system of
adapted coordinates $x^\alpha$, every other surface lying
sufficiently close to ${\cal S}$ may be uniquely described by two
functions $x^0= f^0(u^2,u^3)$ and $x^1 =f^1(u^2,u^3)$ living on
$S^2$.
Any geometric
condition imposed on a topological sphere may, therefore, be
translated into a condition for these two functions. Especially, we are going
to formulate these conditions in terms of two scalar functions,
namely:
the norm of the extrinsic curvature ${\bf k}$ and the two-dimensional divergence of the extrinsic torsion
${\bf t}$. However, mimicking strictly the situation in the flat
Minkowski space, where the round spheres satisfy two conditions:
1) $\|{\bf k}\|-  \left\langle  \|{\bf k}\|  \right\rangle = 0$ and\footnote{By
$ \left\langle f \right\rangle$ we denote the mean value of the function $f$ over ${\cal
S}$.}
2) $\div {\bf t} = 0$, does
not lead to a satisfactory formulation of the problem  Indeed, the above conditions are too restrictive as it will be
seen in the sequel. The main problem arises from the fact that
these conditions imply two (elliptic) differential equations for
$f^a$
which, in general, do not admit an eight-parameter family of
solutions. It turns out that the correct relaxation of these
conditions, valid not only in the Minkowski geometry but also in a
generic geometry $g$, consists in
vanishing a certain projection of the above two
characteristic functions. We define these projections in the next subsection.

\subsection{Projections to eigenspaces}

Let $\overset{\circ}{\Delta} = \overset{\circ}{\nabla}{^A}
\overset{\circ}{\nabla}_A$ be the standard Laplace-Beltrami
operator on $S^2$ and ${\cal M}^3$ be its eigenspace
associated  to eigenvalue $-2$. Then ${\cal M}^3$ is a three-dimensional space and
spanned by the linear
coordinate functions \eqref{eq1} on ${\mathbb R}^3$ restricted to $S^2$, i.e.
\begin{equation}
\label{eq2b} {\cal M}^3 \! := \lin \{\lambda^1\! , \lambda^2\!  ,\lambda^3 \}
\! := \lin \! \left\{\! \!  \sqrt{   \! \tfrac{3}{4 \pi}} \sin \theta \cos \phi , \! \sqrt{
\! \tfrac{3}{4 \pi}} \sin \theta \sin \phi , \! \sqrt{ \! \tfrac{3}{4 \pi}}
\cos \theta \right\} \! .
\end{equation}
Here,  $\lambda^1, \lambda^2,
\lambda^3$ are spherical harmonics of degree 1 on $S^2$  (cf. \cite{Erd, J1, Tr}) and satisfy the equation
\begin{equation}\label{eigenvalue}
\overset{\circ}{\Delta} \lambda^i = -2 \lambda^i \, .
\end{equation}
Moreover, they are normalized
to form an
orthonormal basis of ${\cal M}^3$ in the
Hilbert space $L^2(S^2,\DO)$ with respect to the standard measure
\begin{equation}
\label{measure}
\DO := \sqrt{\det \overset{\circ}{\eta}} \ \mbox{\rm d}^2u \, .
\end{equation}
By $l^i = \left( F^{-1} \right)^*
(\lambda^i)$ we get the corresponding eigenfunctions of $\Delta:= \overset{\sigma}{\Delta}$ on
${\cal S} = im(F)$ and the corresponding space
\begin{equation}\label{linfctn}
   {\cal L}^3({\cal S}) := \left( F^{-1}
   \right)^* ({\cal M}^3)
\end{equation}
where
\begin{align}
\label{rndmetric}
   \sigma &:= \left( F^{-1}
   \right)^*\overset{\circ}{\eta}   &\text{\em (round sphere metric)}, \\
\label{rndmeasure}
  \mbox{\rm d}\mu &:= \sqrt{\det \sigma} \ \mbox{\rm d}^2 u \qquad &\text{\em (round sphere measure)}.
\end{align}
In \cite[Subsection 2.4]{GJKL},  several projections to eigenspaces of $\Delta$ are defined.  They will be applied in  the next subsection.
\begin{defn}
\label{Def0} Let $f \in L^2({\cal S},\mbox{\rm d}\mu)$.
\begin{align}
    f^{\mathbf{m}} &:= P_m(f) := \frac{1}{4\pi}\int_{{\cal S}}
  f \, \mbox{\rm d}\mu \in \lin \{ 1 \}
  &\quad \text{\em  (monopole part)}\, , \label{mon} \\
     f^{\mathbf{d}} &:=
  P_d(f) := \sum_{i=1}^3 l^i \int_{{\cal S}} l^i
  f \, \mbox{\rm d}\mu \in {\cal L}^3({\cal S})
  &\quad \text{\em  (dipole part)} \, , \label{eq1c} \\
     f^{\mathbf{md}} &:=
  P_{md}(f):= P_{m}(f) + P_{d}(f) \in {\cal L}^4({\cal S})
  &\quad \text{\em  (mono-dipole part)} \, , \label{eq1md} \\
     f^{\mathbf{w}} &:=
  (I-P_{md})(f):= f- f^{\mathbf{m}} - f^{\mathbf{d}}
  &\quad \text{\em  (wave part)} \, . \label{eq1w}
\end{align}
\end{defn}
Here we set
\begin{equation}
\label{eq2c} {\cal L}^4({\cal S}) := \lin \{ 1 \}
 \oplus {\cal L}^3({\cal S})
 = \lin \{1,l^1, l^2 ,l^3 \} = \left( F^{-1} \right)^* ({\cal M}^4) \, .
\end{equation}
with ${\cal M}^4 := \lin \{ 1 \} \oplus {\cal M}^3 = \lin \{1,\lambda^1, \lambda^2 , \lambda^3 \}$ being the collection  of all affine functions on $\mathbb{R}^3$
restricted to $S^2$.

All these expressions along with the linear functions $l^i = \left( F^{-1} \right)^*
(\lambda^i)$ are intrinsically defined on ${\cal S} \subset M$ if we use the privileged family of {\em equilibrated,
conformally spherical} isomorphisms $F:$
$S^2 \rightarrow M$. Such parametrizations are constructed in \cite[Subsections 2.1--2.3]{GJKL} and
possess the following properties:
\begin{enumerate}
	\item [i)] The metric tensor $s_{AB}$ on ${\cal S}$ is \emph{conformally equivalent}
to
$\overset{\circ}{\eta}_{A B}$, i.e.
 $   s_{AB} = p \cdot \overset{\circ}{\eta}_{A B} \ ,$
with  a (sufficiently smooth) positive function $p$ on ${\cal S}$.
	\item [ii)] The dipole part of the conformal factor $p$ vanishes, i.e.
  $  P_d(p) = 0 \ .$
\end{enumerate}
By \cite[Theorem 1]{GJKL}, each topological sphere ${\cal S} \subset M$ admits a unique (up to rotations) equilibrated spherical system.

\begin{rem}
\label{Rem1c}
Let us formulate conditions i) and ii) for the isomorphism $F$ in different ways:
Denoting the {\em trace} $\tr M$ and the {\em traceless part} $\t0 M_{A B}$ of a covariant tensor $M=( M_{A B} )$ defined on $S^2$ by
\begin{equation}
\label{eq7} \tr M := \overset{\circ}{\eta}{}^{A B} M_{A B} \quad \text{and} \quad
\t0 M_{A B} := M_{A B} - \tfrac{1}{2} \overset{\circ}{\eta}{}_{A B} \tr M \ ,
\end{equation}
respectively, we obtain $ p= \tfrac{1}{2} \tr s $ in i). Hence, this condition
is equivalent to
\begin{equation}\label{conf1}
    \t0 s_{AB} = 0 \ .
\end{equation}
Since all coefficients in $  P_d(p) = 0$ have to vanish, the condition ii) can be equivalently rewritten in the form:
\begin{equation}\label{bary}
{\bf X} := \left(  \left\langle l^1 \right\rangle, \left\langle l^2 \right\rangle, \left\langle l^3 \right\rangle \right) =0 \qquad \text{\em  (equilibration condition)} \ .
\end{equation}
As usual, the mean value of a function $h$ on ${\cal S}$ is given by
\begin{equation}\label{average}
     \left\langle h \right\rangle := \left( \int_S \sqrt{\det s}\ \mbox{\rm d}^2 u \right)^{-1}
    \cdot \int_S h \sqrt{\det s}\ \mbox{\rm d}^2u  \ .
\end{equation}
We want to point out that the mean values of $l^i$ with respect to another measure vanish, too. These relations can be expressed by
$ \left\langle p^{-1} l^i  \right\rangle =0$ which follows from (\ref{measure}), (\ref{rndmetric}), (\ref{rndmeasure}), (\ref{conf1}), and
\begin{equation}
\label{mean}
P_{m}(l^i) = \frac{1}{4\pi} \int_{\cal S} l^i \, \mbox{\rm d}\mu = \frac{1}{4\pi} \int_{S^2} \lambda^i \DO =0
\end{equation}
\end{rem}

\subsection{Rigid spheres}

Now, we are ready to formulate our geometric conditions which we
are going to impose on two geometric scalar objects of a topological sphere ${\cal S}$: \\
1) the length $k:=\|{\bf k}\| = \sqrt{k^a s_{ab} k^b} $ of the extrinsic
curvature vector, and \\
2) the ``divergence'' of the extrinsic torsion
\begin{equation}\label{q-div-t}
    {\bf q} = \div {\bf t}:=
    \partial_A\left(\sqrt{\det s}\
    s^{AB} \ t_B \right) \, .
    \end{equation}
We stress that the above quantity does not depend upon the choice
of the metric on ${\cal S}$, within the same conformal class.
Indeed, when ${\tilde s}_{AB} = f s_{AB}$, then we have
$\sqrt{\det {\tilde s}}\ {\tilde s}^{AB} = \sqrt{\det s}\ s^{AB}$.
Hence,
we may replace the induced metric $s$ in the above formula by the
round sphere metric $\sigma$. The quantity ${\bf q}$ is a scalar
density which is converted into a scalar by
\begin{equation}\label{scalarq}
    q := \frac {1}{\sqrt{\det \sigma}} \ {\bf q} \, ,\quad \text{i.e.}
\quad q= \overset{\sigma}{\nabla}_A t^A \, .
\end{equation}


\begin{defn}\label{rigid}
A topological sphere ${\cal S} \subset M$ will be called a
{\em rigid sphere} if it possesses an equilibrated spherical parameterization $F$ such that the curvature vector ${\bf k}$ of ${\cal S}$ is space-like and that its both scalars $k$ and $q$ contain only the mono-dipole
part, i.e.~if ${\cal S}=F(S^2)$, where \eqref{conf1}, \eqref{bary} and the two following equations are satisfied:
\begin{align}
  k^{\mathbf{w}}=(I-P_{md}) \, k &= 0 &\quad \text{\em  (curvature condition)}\, , \label{md-k}\\
  q^{\mathbf{w}}=(I-P_{md}) \, q &= 0 &\quad \text{\em  (torsion condition)}\, . \label{md-q}
\end{align}
\end{defn}

\begin{rem}
\label{Rem1}  Observe that the torsion scalar $q$ satisfies the
identity
$
  P_m (q) = \int_{\cal S} q \, \mbox{\rm d}\mu =0 \ ,
$
as a consequence of the definition \eqref{scalarq}. This implies
that the operator $P_{md}$ in \eqref{md-q} can be
replaced by $P_d$. Moreover, \eqref{md-k} and \eqref{md-q} are equivalent to
$$
  k={\| {\bf k} \| } \in {\cal M}^4 \  \text{ and} \quad
  q=\overset{\sigma}{\nabla}_A t^A \in {\cal M}^3 \ ,
$$
respectively.\footnote{Hence, Definition \ref{rigid} coincides with \cite[Definition 4]{GJKL}.}
These relations can be considered as
``gauge conditions'' for extrinsic curvature and torsion.
They are discussed for
the linear case in \cite[Appendix A]{J2}.
\end{rem}

Both conditions, \eqref{md-k} and \eqref{md-q}, are fulfilled by the
eight-parameter family of round spheres in Minkowski space. We are
going to prove in the sequel that these are the only solutions
and, moreover, that \eqref{md-k} and \eqref{md-q} admit always an
eight-parameter family of solutions if the space-time metric does
not differ too much from the flat Minkowski metric.

\section{The linearized problem}


\subsection{Minkowski case}

In case of the Minkowski metric $g=\eta,$ the solution to our
nonlinear problem \eqref{conf1}, \eqref{bary}, \eqref{md-k},
and \eqref{md-q} is obvious since the round spheres
$S=S_{T,R}$ are rigid for all $T\in {\mathbb R}$ and
$R>0.$ To see this analytically, we set
\[
\xi = \xi (t,r,\theta , \phi) = ( t , r \sin \theta \cos \phi ,
r \sin \theta \sin \phi , r \cos \theta ) \, ,
\]
where $t,r,\theta , \phi$ vary within the ranges given by
(\ref{eq2}).  Then, for fixed $T, R,$ it is easy to check that
\begin{align}
\label{eq9}
  &\text{the transformation} \quad x (\theta , \phi)
  = \xi(T,R,\theta , \phi) \\
\intertext{yields an equilibrated spherical coordinate system. Moreover,}
\label{eq9a}
  &\text{the curvatures } \quad \kappa^a := {\eta}^{A B}
  \Gamma^a_{A B}
\end{align}
solve equations (\ref{md-k}) and (\ref{md-q}) at $x \in S.$
Here, $\Gamma^{\alpha}_{\beta \gamma} $ are the Christoffel symbols
with respect to the basis
${\boldsymbol
  \beta}_{\alpha} := \partial_{\alpha} \xi$ where
\begin{equation}
\label{eq11}
\begin{split}
  \partial_0 \xi & = \tfrac{\partial \xi}{\partial t} = (1,0,0,0) \ , \\
  \partial_1 \xi & = \tfrac{\partial \xi}{\partial r} =
  (0,\sin \theta \cos \phi , \sin \theta \sin \phi, \cos \theta) \ , \\
  \partial_2 \xi & = \tfrac{\partial \xi}{\partial \theta} = (0,R \cos
  \theta \cos \phi , R \cos \theta \sin \phi, -R \sin \theta)
  \ , \\
  \partial_3 \xi & = \tfrac{\partial \xi}{\partial \phi} = (0,-R \sin
  \theta \sin \phi , R \sin \theta \cos \phi ,0)
\end{split}
\end{equation}
and $\eta({\boldsymbol \beta}_{\alpha}, {\boldsymbol \beta}_{\gamma}) =
\eta_{\alpha \gamma} $. Hence, the induced metric of $S$
equals to
${\eta}_{A B} = R^2 \overset{\circ}{\eta}_{AB}.$ Recall that the
nonvanishing components of the Christoffel symbols on $S$ are
\begin{gather}
\label{eq12}
\begin{split}
  \Gamma^1_{22} &= \Gamma^{r}_{\theta \theta} = - R \ , \
  \Gamma^1_{33} = \Gamma^{r}_{\phi \phi}
  = - R \sin^2 \theta \ , \\
  \Gamma^2_{12} &= \Gamma^{\theta}_{r \theta} = 1/R \ , \
  \Gamma^3_{13} = \Gamma^{\phi}_{r \phi}
  = 1/R \ , \\
  \Gamma^2_{33} &= \Gamma^{\theta}_{\phi \phi} = -\sin \theta \cos
  \theta \ , \ \Gamma^3_{23} = \Gamma^{\phi}_{\theta \phi}
  =  \cot \theta . \\
\end{split}
\intertext{Therefore, we have the relations:} \label{eq12a}
\Gamma^1_{AB} = -R \overset{\circ}{\eta}_{A B} \, , \
\Gamma^A_{1B} = \tfrac{1}{R} \delta^A_B \, .
\end{gather}
By (\ref{eq6a}), the coordinates ${\kappa}^a$ of the extrinsic curvature vector of $S$ are given and
\begin{equation}
\label{eq13}
{\kappa}^0 = 0 \ , \ {\kappa}^1 = - \tfrac{2}{R} \ ,
\ \sqrt{{\kappa}^a {\eta}_{ab} {\kappa}^b} = \tfrac{2}{R} \ , \
\nabla_A {\kappa}^a =0 \ .
\end{equation}
\defnref{Def1}
implies that, in this case, the first
orthonormal vector of $S$ equals to ${\boldsymbol \beta}_1$,
the second one equals to ${\boldsymbol \beta}_0$, and  the extrinsic torsion
vector $\bf t $ vanishes. Consequently, equations (\ref{md-k}) and
(\ref{md-q}) are fulfilled
(cf. \cite[Example 3]{GJKL}).

To solve the nonlinear problem \eqref{conf1}, \eqref{bary},
\eqref{md-k}, and \eqref{md-q} in general case of a spacetime manifold
$M$, equipped with the metric tensor $g$, we apply the
implicit function theorem (cf., e.g., \cite[Vol.~1]{Z}). For this
purpose, it is useful to represent these equations in the compact
form
\begin{equation}
\label{eq14} {\cal F} (x,g,T,R) := \left(\t0 s_{AB}, {\bf X},
k^{\mathbf{w}} , q^{\mathbf{w}} \right) = 0 \ .
\end{equation}

%
%

Here, $x=x (\theta , \phi)$, $(\theta , \phi) \in [ 0, \pi ] \times
[0,2 \pi)$, encodes our unknown, i.e.~a mapping
$F : S^2 \rightarrow
M $. In a local coordinate system ($x^\alpha$),  $F$ is described by four functions $x^\alpha
= x^\alpha (\theta , \phi)$ which belong to an appropriate
function space specified later. The nonlinear
operator ${\cal F}$ acting on this space represents our
conditions which ensure that the resulting image $F(S^2)
\subset M$ is, indeed, a rigid sphere in the sense of Definition \ref{rigid}.
We are going to
solve the system of equations in \eqref{eq14} separately, for any given values of the
parameters $T$ and  $R > 0$
(cf.~Subsection 5.1).

\begin{rem}
\label{Rem2} The operator equation (\ref{eq14}) consists of three
algebraic equations ${\bf X}=0$ and of four partial differential
equations for the unknowns $x^{\alpha}, \alpha = 0,
\ldots ,3$, defined on the standard sphere $S^2 \subset {\mathbb
R}^3$. In particular, the two equations: $\t0 s_{AB} = 0$ are of the
first differential order. The curvature equation $k^{\mathbf{w}} = 0$
is of the second order, whereas the torsion equation $q^{\mathbf{w}}
= 0$ is of the fourth one.
\end{rem}


\subsection{Linearization of the conformity and equilibration conditions}

We will linearize the operator ${\cal F}$ in (\ref{eq14}) with
respect to the variable $x$ at the point $(x,g,T,R) :=
(\xi,\eta,T,R)$, where $\xi =\xi(T,R,\theta , \phi)$ is given by
formula \eqref{eq9} for an arbitrary value of $T$ and $\eta$ is
the flat Minkowski metric. Setting $x = \xi + \delta x$,
we derive explicit expressions of the (partial) Fr\'echet
derivative
\[
{\cal D}_x {\cal F} (\xi,\eta,T,R) [\delta x] \ .
\]
It is convenient to represent all the terms at a point $x=x(
\theta , \phi) \in {\cal S} = {\cal S} (g,T,R),$ with respect to
the given basis ${\boldsymbol \beta}_0 (\xi),{\boldsymbol \beta}_1
(\xi), {\boldsymbol \beta}_2 (\xi), {\boldsymbol \beta}_3 (\xi)$,
${\boldsymbol \beta}_{\alpha} := \partial_{\alpha} \xi$, $\xi \in
S$. In particular, for $\delta x$ we make the ansatz
\begin{equation}
\label{eq25} \delta x = \psi^{\alpha} {\boldsymbol \beta}_{\alpha}, \
\text{i.e.} \ x = \xi + \psi^{\alpha} (\theta , \phi) {\boldsymbol \beta}_{\alpha}
(\xi) \ .
\end{equation}

Hence, from (\ref{eq25}), (\ref{eq12}), and (\ref{eq12a}), we get the following expression of the basis vectors $ {\bf b}_2, {\bf b}_3$ of the tangent space $T_x {\cal S}$:
\begin{equation}
\label{eq24}
\begin{split}
{\bf b}_A (x)&= \partial_A x = {\boldsymbol \beta}_A (\xi)
+ {\nabla}_A \psi^{\alpha} {\boldsymbol \beta}_{\alpha} (\xi) \\
 &= \left( 1+\frac{\psi^1}{R} \right) {\boldsymbol \beta}_A  +
\overset{\circ}{\nabla}_A \psi^a {\boldsymbol \beta}_a +
\overset{\circ}{\nabla}_A \psi^B {\boldsymbol \beta}_B - R  \overset{\circ}{\eta}_{A B} \psi^B {\boldsymbol \beta}_1 \ .
\end{split}
\end{equation}
In the last formula, we have introduced the two-dimensional
covariant derivative $\overset{\circ}{\nabla}_A$ which we apply to
a vector field $(\omega^0,\omega^1,\omega^2,\omega^3)$ defined on
$S^2$ in a splitted way
\begin{equation}
\label{eq24a} \overset{\circ}{\nabla}_A \omega^{B} := \partial_A
\omega^B + {\Gamma}^B_{AC} \omega^C \ , \
\overset{\circ}{\nabla}_A \omega^a := \partial_A \omega^a \ .
\end{equation}
This means that the tangent part $(\omega^A)$ of the above field
is treated as a genuine vector on $S^2$, whereas its transversal
components $\omega^0$ and $\omega^1$ are treated as ``scalars''.
The reason for this choice is that the ``tangent-transversal''
and the ``transversal-transversal'' components of the
four-dimensional connection will be explicitly described in terms
of the extrinsic curvature ${\bf k}$ and the extrinsic torsion~${\bf t}.$ Here, we
take into account only its remaining ``tangent-tangent'' part. \\
Let us notice that the Christoffel symbols ${\Gamma}^C_{AB}$
listed in (\ref{eq12}) are the same as for the standard sphere
$S^2$ equipped by the canonical metric $\overset{\circ}{\eta}_{A
B}.$ Relation (\ref{eq24}) implies the representation of the
induced metric $s$ given by (\ref{eq4}):
\begin{align}
s_{A B} &= \eta ({\bf b}_A ,{\bf b}_B ) = \frac {\partial x^{\alpha}}{\partial u^A} \
  \eta_{\alpha \beta} \ \frac {\partial x^{\beta}}{\partial u^B}  \nonumber \\
\label{eq11b}
&= R^2 \left( \overset{\circ}{\eta}_{A B}
+ 2 \frac{\psi^1}{R} \overset{\circ}{\eta}_{A B} +
\overset{\circ}{\nabla}_A \psi_B + \overset{\circ}{\nabla}_B \psi_A \right)+ O(\delta^2) \ .
\end{align}
Here and in the following, $O(\delta^2)$ denotes the collection of
all terms quadratic in $ \delta x $ or their derivatives.
If  the tensor $(s_{A B})$ in (\ref{eq11b}) is associated to a point
$x \in {\cal S}$, we obtain
\begin{align}
s_{A B} (x) &= s_{A B} (\xi) + \psi^{\alpha} \partial_{\alpha} s_{A B} (\xi) + O(\delta^2)
= s_{A B} (\xi) + \psi^{\alpha} \partial_{\alpha} \eta_{A B} (\xi) + O(\delta^2) \nonumber \\
\label{eq11a}
&= R^2 \left( \overset{\circ}{\eta}_{A B}
+ 4 \frac{\psi^1}{R} \overset{\circ}{\eta}_{A B} + \psi^C \partial_C \overset{\circ}{\eta}_{A B} +
\overset{\circ}{\nabla}_A \psi_B + \overset{\circ}{\nabla}_B \psi_A \right)+
O(\delta^2) \, ,
\end{align}
and immediately deduce the expansion of $\t0 s_{AB}$ in  conformity condition \eqref{conf1}.

\begin{prop}
\label{Prop3}
The linearization (in $x$) of $ \ \t0 s_{AB} (x)= s_{AB}(x)- \tfrac{1}{2}\left( \overset{\circ}{\eta}_{A B} \overset{\circ}{\eta}{}^{CD} s_{CD} \right) (x)$,
$x \in {\cal S}$
equals to
\[
 R^2 \left( \overset{\circ}{\nabla}_A \psi_B + \overset{\circ}{\nabla}_B \psi_A
- \overset{\circ}{\eta}_{A B} \overset{\circ}{\nabla}{_C} \psi^C \right) \ .
\]
\end{prop}

\begin{proof}
By (\ref{eq11a}) and $ \overset{\circ}{\eta}_{A B} (x) = \overset{\circ}{\eta}_{A B} (\xi) + \psi^C \partial_C \overset{\circ}{\eta}_{A B} (\xi) + O(\delta^2)$,
all terms in the above expansion are obvious since
$\t0 \partial_C \overset{\circ}{\eta}_{A B} (\xi) $ vanishes up to first order.
To see this, observe

\parbox{0.9\textwidth}{
\[
\overset{\circ}{\eta}{}^{AB}(x) = \overset{\circ}{\eta}{}^{AB}(\xi) - \left(  \overset{\circ}{\eta}{}^{AD} \overset{\circ}{\eta}{}^{BE} \psi^C \partial_C \overset{\circ}{\eta}_{DE} \right) (\xi) + O(\delta^2) \ .
\]}
\end{proof}

To linearize the average $ \left\langle l^i \right\rangle$ in the equilibration condition \eqref{bary},
 we use the abbreviation $ \left\langle h \right\rangle^{\circ}$ for the average of a function $h$ defined on
$S^2$, i.e.
\begin{equation}\label{staverage}
     \left\langle h \right\rangle^{\circ} := \frac{1}{4\pi} \int_{S^2}  h \DO \ .
\end{equation}

\begin{prop}
\label{Prop4}
The linearization (in $x$) of $ \left\langle l^i \right\rangle$
equals to
\[
  \left\langle  \lambda^i ( \tfrac{4}{R} \psi^1 + \partial_A \psi^A  )   \right\rangle^{\circ} \, .
\]
\end{prop}

\begin{proof}
We first use (\ref{eq11a})
to get
\begin{align*}
\det s (x) &=  R^4 \left( \left( 1 + 4 \frac{\psi^1}{R} \right) ^2 \det\overset{\circ}{\eta}  +
\psi^A \partial_A \det\overset{\circ}{\eta}
+ 2 \overset{\circ}{\eta}_{22} \overset{\circ}{\nabla}_3 \psi_3 + 2 \overset{\circ}{\eta}_{33} \overset{\circ}{\nabla}_2 \psi_2 \right) (\xi) + O(\delta^2) \\
&= R^4 \left( 1 + 8 \frac{\psi^1}{R} + 2 \psi^A {\Gamma}^{B}_{AB} +
2 \overset{\circ}{\nabla}{_A} \psi^A \right) \det\overset{\circ}{\eta}
+ O(\delta^2)\, ,
\end{align*}
since $ \overset{\circ}{\eta}_{22}= \overset{\circ}{\eta}{}^{33} \det\overset{\circ}{\eta}$, $ \overset{\circ}{\eta}_{33}= \overset{\circ}{\eta}{}^{22} \det\overset{\circ}{\eta}$, and
$\partial_A \det\overset{\circ}{\eta} = 2  {\Gamma}^{B}_{AB} \det\overset{\circ}{\eta}$.
Expansion of the square root yields
\[
\sqrt{\det s (x)} =  R^2 \left( 1+ 4 \frac{\psi^1}{R} +
\partial_A \psi^A \right) \sqrt{\det\overset{\circ}{\eta}} + O(\delta^2) \, ,
\]
and hence
\[
 \int_{\cal S} l^i \sqrt{\det s}\ \mbox{\rm d}^2 u = R^2 \int_{S^2} \lambda^i \left( 4 \frac{\psi^1}{R}
 + \partial_A \psi^A \right)  \sqrt{\det \overset{\circ}{\eta}} \ \mbox{\rm d}^2u + O(\delta^2) \ .
\]
In the derivation of the last formula, relations (\ref{rndmetric}),
(\ref{rndmeasure}), and (\ref{mean})
have to be observed. Finally, by (\ref{average}) and (\ref{staverage}),
the assertion follows.
\end{proof}

\subsection{Linearization of the curvature and torsion conditions}

To get the linearized expression of the extrinsic curvature vector (\ref{k})
at $g=\eta$, we use adapted coordinates and the transformation rule for the Christoffel
symbols ${  \Gamma}^{\alpha}_{\beta \gamma} (\xi)$, if we pass from the natural
basis ${\boldsymbol \beta}_{\alpha}(\xi)$ to the basis ${\bf  b}_{\alpha}(x)$.
Let
\begin{equation}
\label{eq16}
{\bf b}_{\alpha} = (\delta_{\alpha}^{\gamma} +
b_{\alpha}^{\gamma}) {\boldsymbol \beta}_{\gamma} \quad
\text{and} \quad
{\boldsymbol \beta}_{\alpha} = (\delta_{\alpha}^{\gamma} +
\beta_{\alpha}^{\gamma}) {\bf b}_{\gamma}
\end{equation}
then $b_A^{\gamma}$ follows from (\ref{eq24}). To calculate $b_a^{\gamma}$,
 we recall (\ref{b_a})
and notice
\begin{equation}
    {\bf b}_a = {\boldsymbol \beta}_a^\perp = {\boldsymbol \beta}_a - \eta ( {\boldsymbol \beta}_a, {\bf b}_B ) s^{BC} {\bf b}_C
= {\boldsymbol \beta}_a - {\nabla}_B \psi_a \eta^{BC} {\boldsymbol \beta}_C + O(\delta^2)
		\label{eq12b}
\end{equation}
by (\ref{eq24}), (\ref{eq11a}). Hence,
\begin{align*}
        b_A^C&= \overset{\circ}{\nabla}_A \psi^C + \frac{\psi^1}{R} \delta_A^C \ , \quad
        b_A^c = \overset{\circ}{\nabla}_A \psi^c - R \overset{\circ}{\eta}_{AB} \psi^B \delta_1^c \ , \\
        b_a^C&= - \frac{1}{R^2}\overset{\circ}{\nabla}{^C} \psi^b {\eta}_{ab} + \frac{\psi^C}{R} \delta_a^1 \ , \quad   b_a^c=0   \, , \
         \text{and} \quad \beta_{\alpha}^{\gamma} = - b_{\alpha}^{\gamma} +
O(\delta^2)
\end{align*}
in (\ref{eq16}). Note $(\delta_{\alpha}^{\mu} +
b_{\alpha}^{\mu})(\delta_{\mu}^{\gamma} + \beta_{\mu}^{\gamma}) =
\delta_{\alpha}^{\gamma} .$ If $G^{\alpha}_{\beta \gamma}$ are the
Christoffel symbols of the metric $\eta$ with respect to the basis
${\bf b}_{\alpha}$ associated to $\xi \in S$, we get
\[
  G^{\alpha}_{\beta \gamma} = \left( \delta^{\alpha}_{\lambda} +
    \beta^{\alpha}_{\lambda} \right) \left( \delta_{\beta}^{\mu} +
    b_{\beta}^{\mu} \right) \left( \delta_{\gamma}^{\nu} +
    b_{\gamma}^{\nu} \right) \Gamma^{\lambda}_{\mu \nu} +
  \left( \delta^{\alpha}_{\nu} + \beta^{\alpha}_{\nu} \right)
\partial_{\beta} \left( \delta_{\gamma}^{\nu} + b_{\gamma}^{\nu}
  \right) \ .
\]
Especially
\begin{align}
G^{a}_{A B} (\xi)
  &= { \Gamma}^{a}_{AB} + \beta^{a}_C {\Gamma}^{C}_{AB} + b_A^C {\Gamma}^a_{CB}
  + b_B^C {\Gamma}^a_{AC} + \partial_A b^a_B + O(\delta^2) \nonumber \\
\label{eq25b} &= \left( 1 + 2 \frac{\psi^1}{R} \right){\Gamma}^{a}_{AB} (\xi) +
\overset{\circ}{\nabla}_A \overset{\circ}{\nabla}_B \psi^a - R \left( 2 \overset{\circ}{\nabla}_A \psi_B + \overset{\circ}{\nabla}_B \psi_A \right) \delta_1^a + O(\delta^2) \ , \\
G^{a}_{A c}(\xi)&=  \beta^{a}_B {\Gamma}^{B}_{Ac} (\xi) + b_c^B {\Gamma}^a_{AB} (\xi) + O(\delta^2) \nonumber \\
 &= \left(  -\overset{\circ}{\nabla}_B \psi^a + R \psi_B \delta_1^a \right) {\Gamma}^{B}_{Ac}  + \left( - \frac{1}{R^2}\overset{\circ}{\nabla}{^B} \psi_c + \frac{\psi^B}{R} \delta_c^1 \right) {\Gamma}^a_{AB} + O(\delta^2) \nonumber \\
\label{eq25c} &= \frac{1}{R} \left( \overset{\circ}{\nabla}{_A} \psi_c \delta_1^a - \overset{\circ}{\nabla}_A \psi^a \delta_c^1  \right) + O(\delta^2) \ .
\end{align}
Here, the representation of $\Gamma^{\alpha}_{\beta \gamma}$ in
(\ref{eq12}), (\ref{eq12a}) and the definition (\ref{eq24a}) have been exploited.
Finally, we deduce
\begin{prop}
\label{Prop1}
The linearization (in $x$) of $k(x) = \sqrt{k^a s_{ab} k^b (x)}$
equals to
\[
- \frac{1}{R^2} ( \overset{\circ}{\Delta} +2 ) \psi^1 +
\frac{1}{R} \overset{\circ}{\nabla}{_A} \psi^A \, .
\]
\end{prop}

\begin{proof}
Let $x \in {\cal S}$. Then $G^{a}_{A B} (x)$ expands to
\begin{align*}
 G^{a}_{A B} (x) &= G^{a}_{A B} (\xi) + \psi^{\alpha} \partial_{\alpha} G^{a}_{A B} (\xi) + O(\delta^2)
= G^{a}_{A B} + \psi^{\alpha} \partial_{\alpha} \Gamma^{a}_{A B} + O(\delta^2) \\
&= \left( 1 + 3 \frac{\psi^1}{R} \right){\Gamma}^{a}_{AB}  - R \left( \psi^C \partial_C \overset{\circ}{\eta}_{A B} + 2 \overset{\circ}{\nabla}_A \psi_B + \overset{\circ}{\nabla}_B \psi_A \right) \delta_1^a
+ \overset{\circ}{\nabla}_A \overset{\circ}{\nabla}_B \psi^a + O(\delta^2) \\
&= - \frac{1}{R} \left( s_{A B} (x) - R \psi^1 \overset{\circ}{\eta}_{AB} + R^2 \overset{\circ}{\nabla}_A \psi_B \right) \delta_1^a +
\overset{\circ}{\nabla}_A \overset{\circ}{\nabla}_B \psi^a + O(\delta^2)
\end{align*}
by (\ref{eq12a}), (\ref{eq11a}), and (\ref{eq25b}). Use the expression of the inverse $ \left( s^{A B} \right) $ to matrix $ \left( s_{A B} \right) $ in formula (\ref{eq6a}) to obtain
\begin{equation}
\label{eq23}
\begin{split}
  k^0 (x)&
= \dfrac{1}{R^2} \overset{\circ}{\Delta} \psi^0 + O(\delta^2) \ , \\
  k^1 (x)&
= -  \dfrac{2}{R} + \dfrac{1}{R^2} \left( \overset{\circ}{\Delta} +2 \right) \psi^1
- \frac{1}{R} \overset{\circ}{\nabla}_A \psi^A + O(\delta^2) \, .
\end{split}
\end{equation}
Since $s_{ac}= \eta ( {\bf b}_a , {\bf b}_c)= {\eta}_{ac} + O(\delta^2) $
by (\ref{s}) and (\ref{eq12b}), we get
\begin{gather}
\label{eq26}
s_{ac} (x) =  {\eta}_{ac} (\xi) + \psi^{\alpha} \partial_{\alpha} s_{ac} (\xi) + O(\delta^2)
= {\eta}_{ac} + O(\delta^2) \\
\label{eq27}
\text{and} \quad k^a s_{ab} k^b (x) = \frac{4}{R^2} \left( 1 - \frac{1}{R} ( \overset{\circ}{\Delta}
+2 ) \psi^1 + \overset{\circ}{\nabla}_A \psi^A \right) + O(\delta^2) \  .
\end{gather}
Taylor expansion of the square root function gives the assertion.
\end{proof}

\begin{prop}
\label{Prop2}
The linearization (in $x$) of $q(x) = \overset{\sigma}{\nabla}_A t{^A}
(x)$, $x \in {\cal S}$,
equals to
\[
 -\frac{1}{2R} \overset{\circ}{\Delta} ( \overset{\circ}{\Delta} +2 ) \psi^0 \ .
\]
\end{prop}

\begin{proof}
Recall the definitions (\ref{m}) and (\ref{n}) of the first and second orthonormal vectors ${\bf  n}$, ${\bf  m}$ and the relation $ \|{\bf k}\| = |k^1| + O(\delta^2)$ by (\ref{eq23}), (\ref{eq27}) to derive
\begin{align*}
 m^0 &= - n^1 +  O(\delta^2) = 1 +  O(\delta^2) \ , \
m^1= - n^0 + O(\delta^2) = - \frac{1}{2R} \overset{\circ}{\Delta} \psi^0 +  O(\delta^2) \ , \\
\overset{g}{\nabla}_A {\bf  n} &= \overset{\eta}{\nabla}_A \left( n^a {\bf b}_a \right) =
\left( \partial_A n^a + G^a_{A b} n^b \right) {\bf b}_a
= \frac{1}{2R} \left( \overset{\circ}{\nabla}_A \overset{\circ}{\Delta} \psi^0
+ 2 \overset{\circ}{\nabla}_A \psi^0 \right) {\boldsymbol \beta}_0 +  O(\delta^2) \ .
\end{align*}
Here, the representations of the basis ${\bf b}_a$ and of the associated Christoffel symbols $G^a_{A b}$ in
(\ref{eq12b}) and (\ref{eq25c}), respectively, are used. From (\ref{eq6b})
\begin{align*}
 q &= \overset{\sigma}{\nabla}_A \left( \overset{\circ}{\eta}^{A B} \eta ( {\bf  m}, \overset{g}{\nabla}_B {\bf  n}) \right) \\
&= - \frac{1}{2R} \overset{\circ}{\nabla}^A m^0 \left( \overset{\circ}{\nabla}_A \overset{\circ}{\Delta} \psi^0
+ 2 \overset{\circ}{\nabla}_A \psi^0 \right) +  O(\delta^2)
= -\frac{1}{2R} \overset{\circ}{\Delta} ( \overset{\circ}{\Delta} +2 ) \psi^0
+  O(\delta^2)
\end{align*}
follows. Omitting the terms in $O(\delta^2)$, we find the desired
linearization.
\end{proof}

\section{Analysis of the linearized operator}

\subsection{The linearized operator}

The propositions derived in the previous section immediately yield the Fr\'echet
derivative ${\cal A} := {\cal D}_{x} {\cal F} ( \xi, \eta , T, R)$
of the nonlinear operator ${\cal F}$ defined in (\ref{eq14}). It is convenient to represent the operator ${\cal A}$ in the form
\begin{equation}
\label{eq31b}
{\cal A}= ({\cal A}_{A B} , {\cal A}^i , {\cal A}_{curv} , {\cal A}_{tors} )\, , \quad A,B=2,3; \ i=1,2,3\, ,
\end{equation}
where its components correspond to those ones of ${\cal F}$ in \eqref{eq14}.

\begin{thm}
\label{Thm1} Let  the parameters $T\in {\mathbb R}$ and $R>0$ be
fixed and let $ \overset{\circ}P_{md}$ ($ \overset{\circ}P_{m}$, resp. $ \overset{\circ}P_{d}$) denote
 the projections $ P_{md}$ ($ P_{m}$, resp. $ P_{d}$) transformed to $S^2$. Then
\begin{align*}
\text{\bf i)} \quad  & {\cal A}_{A B} [\psi]&&=
R^2 \left( \overset{\circ}{\nabla}_A \psi_B + \overset{\circ}{\nabla}_B \psi_A
- \overset{\circ}{\eta}_{A B} \overset{\circ}{\nabla}{_C} \psi^C \right) \quad \text{on }
S^2 \ ,& \\
\text{\bf ii)} \quad & {\cal A}^i [\psi]&&=
 \left\langle  \lambda^i ( \tfrac{4}{R} \psi^1 + \partial_A \psi^A  )   \right\rangle^{\circ}  \ ,& \\
\text{\bf iii)} \quad & {\cal A}_{curv} [\psi]&&=
- \frac{1}{R^2} ( \overset{\circ}{\Delta} +2 - 2 \overset{\circ}P_{m}) \psi^1 +
\frac{1}{R} (I- \overset{\circ}P_{md}) \, \overset{\circ}{\nabla}{_A} \psi^A \quad \text{on }
S^2 \ ,& \\
\text{\bf iv)} \quad & {\cal A}_{tors} [\psi]&&=
-\frac{1}{2R} \overset{\circ}{\Delta} (\overset{\circ}{\Delta} +2 )\psi^0 \quad \text{on }
S^2 \ .&
\end{align*}
\end{thm}

\begin{proof}
All relations are obvious. Only iii) and iv) require some comments. The linearity of the projection $I-P_{md}$ and
\propref{Prop1} imply the expression
\[
 - \frac{1}{R^2} (I- \overset{\circ}P_{md}) ( \overset{\circ}{\Delta} +2 ) \psi^1 +
\frac{1}{R} (I- \overset{\circ}P_{md}) \, \overset{\circ}{\nabla}{_A} \psi^A
\]
for the linearization of $k^{\mathbf{w}}$. Observe that, by (\ref{mon}) and (\ref{eq1c}),
\begin{equation}
\label{eq21}
\overset{\circ}P_m(f) = \frac{1}{4\pi} \int_{S^2} f \DO \ , \
 \overset{\circ}P_d(f) = \sum_{i=1}^3 \lambda^i \int_{S^2} \lambda^i f \DO
\end{equation}
for a function $f \in L^2(S^2,\DO)$. Since $1$ (resp. $\lambda^i$) is an eigenfunction to the eigenvalue $0$ (resp. $-2$ ) of the operator $\overset{\circ}{\Delta}$ (see \eqref{eigenvalue}), we have
\begin{equation}
\label{eq21a}
\overset{\circ}P_m \overset{\circ}{\Delta} = \overset{\circ}{\Delta} \overset{\circ}P_m = 0 \ , \
\overset{\circ}P_d \overset{\circ}{\Delta} = \overset{\circ}{\Delta} \overset{\circ}P_d = -2 \overset{\circ}P_d \ .
\end{equation}
Analogously, the linearization of $q^{\mathbf{w}}$  follows from \propref{Prop2}.
\end{proof}

In preparation of the next section, we analyze some basic features
of the linear operator ${\cal A}$ given by (\ref{eq31b}). In
particular, the kernel of this linear operator is needed. For this
purpose, let us solve the homogeneous linearized problem ${\cal A}[\psi]=0$ by decoupling the  components.

\subsection{Solution of the homogeneous linearized problem}

To determine the solution of ${\cal A}[\psi]=0$, we start with the last component ${\cal A}_{tors}[\psi]=0$ which is equivalent to
\begin{equation}
\label{eq29} ( \overset{\circ}{\Delta} +2) \overset{\circ}{\Delta} \psi^0 = 0 \ .
\end{equation}
This says that  $ \overset{\circ}{\Delta} \psi^0 $ has
to be an eigenfunction to the eigenvalue $-2$ of the operator
$\overset{\circ}{\Delta}$, i.e.~$ \overset{\circ}{\Delta} \psi^0 \in {\cal M}^3 \subset L^2(S^2,\DO)$.
Since $f \in  L^2(S^2,\DO)$ is constant iff
$\overset{\circ}{\Delta} f=0$, we deduce
\begin{equation}
\label{eq34}
\psi^0 \in  {\cal M}^4 \ .
\end{equation}
The same argument is used to solve ${\cal A}_{AB}[\psi]=0$ and yields the tangent part $(\psi^A)$ of the vector field $\psi$. If we apply $\overset{\circ}{\nabla}{^A} $ to
\begin{gather}
\label{eq37}
\overset{\circ}{\nabla}_A \psi_B + \overset{\circ}{\nabla}_B \psi_A
- \overset{\circ}{\eta}_{A B} \overset{\circ}{\nabla}{_C} \psi^C =0 \, ,
\intertext{we get}
\label{eq34a}
 \overset{\circ}{\nabla}{^A} \overset{\circ}{\nabla}_A \psi_B +
\psi_B = 0 \, .
\end{gather}
Here, the covariant derivatives $\overset{\circ}{\nabla}{^A} \overset{\circ}{\nabla}_B $ were commuted and hence the values of curvature tensor $\overset{\circ}{R}_{ABCD}$ as well as the thoses ones of Ricci tensor $\overset{\circ}{R}_{AB}$ on the standard sphere $S^2$ have to be observed ($\overset{\circ}{R}_{AB} = \overset{\circ}{\eta}_{AB}$). Due to Hodge decomposition
we have the representation (cf. \cite[Relation (A.1)]{GJKL}
\begin{equation}
\label{eq32}
\psi_A = \partial_A f^{(1)} + \overset{\circ}{\varepsilon}_A{^B}\partial_B f^{(2)} \, ,
\end{equation}
where $f^{(i)}$, $i=1,2$, are smooth periodic functions and
$\overset{\circ}{\varepsilon}{_{AB}} = \sqrt{\det \overset{\circ}{\eta}} \, E_{AB}$ is the Levi-Civita tensor for $\overset{\circ}{\eta}$ on $S^2$.
Using this in (\ref{eq34a}) and differentiating again,
a short calculation leads to the two equations
\begin{equation}
\label{eq32a}
( \overset{\circ}{\Delta} + 2 ) \overset{\circ}{\Delta} f^{(i)}  = 0 \ .
\end{equation}
They are the same as  (\ref{eq29}) and the analogon of (\ref{eq34}) says
$f^{(i)} \in  {\cal M}^4$. The result for $\psi_A$ can be written symbolically in the form
\begin{equation}
\label{eq33}
\psi_A \in d {\cal M}^3 \oplus *d {\cal M}^3
\end{equation}
according to the decomposition (\ref{eq32}). Let us notice that
$ \overset{\circ}{\nabla}{_A} \psi^A = \overset{\circ}{\Delta} f^{(1)} \in  {\cal M}^4 $. Consequently, the term  $(I- \overset{\circ}P_{md}) \, \overset{\circ}{\nabla}{_A} \psi^A$ vanishes in the equation ${\cal A}_{curv}[\psi]=0$ and
\begin{equation}
\label{eq29a}
( \overset{\circ}{\Delta} +2 - 2 \overset{\circ}P_{m}) \psi^1 = 0
\end{equation}
remains left. Applying $\overset{\circ}{\Delta}$ to the both sides of (\ref{eq29a}), we deduce equation (\ref{eq29}) for $\psi^1$ (instead of $\psi^0$), and hence,
\begin{equation}
\label{eq30}
\psi^1  \in  {\cal M}^4
\end{equation}
by (\ref{eq34}). Here, the first relation of (\ref{eq21a}) has to be observed.


\begin{rem}
\label{Rem5} The above results (\ref{eq34}), (\ref{eq33}), and (\ref{eq30}) show that the solutions $\psi$ of the homogeneous linearized problem
${\cal A}[\psi]=0$ generate a fourteen-parameter family. Each component $\psi^{\alpha} $ is a linear combination of the spherical
harmonics $\lambda^1, \lambda^2 , \lambda^3$, and $1$ or their derivatives. Altogether they contain fourteen real constants independent of $R, T$. But, due to the three linear equations ${\cal A}^i[\psi]=0$, $i=1,2,3$, some of the constants are coupled, as we will see in the sequel.
Hence, there are eleven free parameters in the solution $\psi$. They can be interpreted in the way described below. \\
The three parameters contained in the dipole part of
$\overset{\circ}{\nabla}_A (\overset{\circ}{\varepsilon}{^A}_B \psi^B)$ represent possible rotations of the original sphere. They keep the sphere unchanged. The remaining eight parameters describe its possible deformations. In particular,
the monopole part of  $\psi^1$ describes change of its size due to dilation.
The dipole part of  $\overset{\circ}{\nabla}_A\psi^A$ generates uniquely the dipole part of $\psi^1$. The corresponding three parameters describe possible space translations of the original sphere. The monopole part of $\psi^0$ describes its time translation. Finally, the three parameters contained in the dipole part of $\psi^0$ describe its boost transformations in Minkowski space.
\end{rem}

To solve ${\cal A}^i[\psi]= \left\langle  \lambda^i ( \tfrac{4}{R} \psi^1 + \partial_A \psi^A  )   \right\rangle^{\circ} =0$, we exploit the representation of the radial part $\psi^1$ given by (\ref{eq30}).  Let $\psi^1 = c+ \sum_{i=1}^3 c_i \lambda^i $
with real constants $c,c_1,c_2,c_3$. Then $ \left\langle  \lambda^i \psi^1  \right\rangle^{\circ} = \frac{1}{4\pi} c_i$ since
$ \lambda^1,\lambda^2,\lambda^3 $ form an orthonormal system in $L^2(S^2,\DO)$ and $ \left\langle  \lambda^i  \right\rangle^{\circ} =0$ by
(\ref{mean}). Therefore, only the constant $c$ is free whereas $c_i = - \pi R  \left\langle  \lambda^i \partial_A \psi^A  \right\rangle^{\circ}$. This and the definitions (\ref{staverage}), (\ref{eq21}) imply
\begin{equation}
\label{eq31}
\psi^1 = c - \tfrac{R}{4} \overset{\circ}P_d(\partial_A \psi^A) \ .
\end{equation}

\begin{rem}
\label{REMinhom}
The above calculations also show that, within the subspace ${\cal M}^4$, the function
\begin{equation}
\label{eq31a}
h = c - \tfrac{R}{4} \overset{\circ}P_d(\partial_A \psi^A) + \pi R \sum_{i=1}^3 r_i \lambda^i \ , \
c \in {\mathbb R} \ ,
\end{equation}
is the general solution of the inhomogeneous linear system
\begin{equation}
\label{eq32b}
 \left\langle  \lambda^i ( \tfrac{4}{R} h + \partial_A \psi^A  )   \right\rangle^{\circ} = r_i \ , \ i=1,2,3 \ .
\end{equation}
Here, the constants $r_i$ and the field $(\psi^A)$ are given.
\end{rem}

\subsection{Properties of the linearized operator}

Up to now, all unknowns in the preceding relations
 are assumed to be sufficiently
 smooth functions defined on
$S^2.$ Naturally, we may consider ${\cal A}$ as
a mapping between different Sobolev spaces $H^n := W^{2,n}
(S^2,\DO),$ $n=0,1, \dots$, which are the collections
of all functions $f: S^2 \rightarrow \mathbb R $ possessing
generalized partial derivatives $ D^{\beta}f \in L^2(S^2,\DO)$ for each multiindex $\beta$ with $|\beta| \leq n$.
Note that $H^n $ is a Hilbert space
with scalar product
\[
\left\langle f_1,f_2 \right\rangle_{H^n} = \sum_{|\beta| \leq n} \int_{S^2} D^{\beta}f_1 D^{\beta}f_2 \DO
\quad \text{for } f_1,f_2 \in H^n \, ,
\]
and that the space of smooth functions is dense in each $H^n$ with respect to the corresponding norm. Moreover, $H^{n+2}$ is compactly imbedded in the H{\"o}lder space $C^{n, \nu} (S^2),$ $\nu \in (0,1),$ and
constitutes an algebra of functions (cf., e.g., \cite[Vol.~1, pp 281]{T}).\footnote{Observe that \cite[Theorem 2]{GJKL} on the existence of a unique equilibrated spherical system on a topological sphere $S^2$ is also valid within Sobolev framework (substitute $C^{k, \alpha} (S^2)$ with
$H^{k+2}$, $k \geq 1$).}

In association with $H^n $ we introduce the space $T^n $ of two-fold covariant symmetric traceless tensor fields living on the standard sphere $S^2$
\begin{equation}
\label{eq44}
T^n := \left\lbrace Z=( Z_{A B} ) \left| Z_{A B} =Z_{B A} \in H^n  \text{ for } A,B \in \{2,3 \} , \
\tr Z = 0  \right. \right\rbrace .
\end{equation}
We summarize the essential properties concerning ${\cal A}$ in the following
theorem.
\begin{thm}
\label{Thm2} Let the operator ${\cal A}$ be defined by
(\ref{eq31b}) and let $D({\cal A}),$ $N({\cal A}),$ $R({\cal A})$
denote its domain, null space and range, respectively. \\
If $D({\cal A}) = (H^7)^4 $, then
\begin{align*}
\text{\bf a)} \quad &N({\cal A}) = \{ \psi = (\psi^{\alpha}) |
\ \psi^0 \in {\cal M}^4  , \ \psi_A \in d {\cal M}^3 \oplus *d {\cal M}^3 ,
\ \psi^1 \in - \tfrac{R}{4} \overset{\circ}P_d(\partial_A \psi^A) + \lin \{1 \}
\} \, . \\
\text{\bf b)} \quad
&N({\cal A}) \text{ is an eleven-dimensional subspace of } D({\cal A}) \, . \\
\text{\bf c)} \quad &R({\cal A}) = T^6 \times {\mathbb R}^3 \times
(I- \overset{\circ}P_{md}) H^5 \times (I- \overset{\circ}P_{md}) H^3 \, , \\
&\text{where }
(I- \overset{\circ}P_{md}) H^n \text{ is the $L^2$-orthogonal complement of } {\cal M}^4 \text{
in the subspace } \\
&H^n \subset L^2(S^2,\DO) \, . \\
\text{\bf d)} \quad &{\cal A} \text{ is a Fredholm operator of
index 0}.
\end{align*}
\end{thm}

\begin{proof}
  {\em Ad a),b)} The assertions
  follows from (\ref{eq34}), (\ref{eq33}), (\ref{eq31}) and \remref{Rem5}.

  {\em Ad c),d) Step 1.} We follow the lines of the considerations in Subsection 4.2
 which give the solution of the homogeneous problem. The arguments are based on the properties of the operator
$- \overset{\circ}{\Delta}$.  It is non-negative, selfadjoint, and the
Fredholm alternative holds. Hence, the equation
\[
(\overset{\circ}{\Delta} + 2) h = f \quad \left(
\overset{\circ}{\Delta} h = f \right)
\]
possesses a solution $h \in H^{n+2}$ iff $f$ is from $H^n$ and is $L^2$-orthogonal to all
eigenfunctions of $\overset{\circ}{\Delta}$ to the eigenvalue $-2$
(resp. 0), i.e. $L^2$-orthogonal to $\lambda^1, \lambda^2 ,\lambda^3$ (resp. to 1).
This is equivalent to $f \in (I- \overset{\circ}P_{d}) H^n$
(resp. $f \in (I- \overset{\circ}P_{m}) H^n$).
Hence,  the assertions on the component ${\cal A}_{tors}$ are obvious.

{\em Step 2.}
To see that $R({\cal A}_{AB})=T^6$, we refer to \cite[Appendix  A.2]{GJKL}. There it is shown that the differential operator defining ${\cal A}_{AB}$ represents an isomorphism between covector fields and symmetric traceless tensors living on $S^2$. Only the slight modifaction of using Sobolev spaces instead H{\"o}lder ones is necessary.


{\em Step 3.} We consider the remaining components ${\cal A}^i$, ${\cal A}_{curv}$ of the operator ${\cal A}$ and solve ${\cal A}^i[\psi]=s_i$, ${\cal A}_{curv}[\psi]=g $ for the radial part
$\psi^1$ provided $s_1,s_2,s_3 \in \mathbb{R}$ and $g \in (I- \overset{\circ}P_{md}) H^5$ are given.
The last relation reduces to
\begin{gather}
\label{eq47}
( \overset{\circ}{\Delta} +2 - 2 \overset{\circ}P_{m}) \psi^1 = f
\intertext{or, equivalently,}
\label{eq46}
\overset{\circ}{\Delta} (\overset{\circ}{\Delta} +2 )\psi^1 = \overset{\circ}{\Delta} f \text{ with }
f \in (I- \overset{\circ}P_{md}) H^5 \ .
\end{gather}
Observe that $\overset{\circ}P_{m} f =0$. Step 1  implies that (\ref{eq46}) has a solution $\psi^1 \in H^7$ iff $ \overset{\circ}{\Delta} f \in (I- \overset{\circ}P_{md}) H^3$. This condition is fulfilled for any $f \in (I- \overset{\circ}P_{md}) H^5 $ since
\[
\overset{\circ}P_{md} \overset{\circ}{\Delta} f
= - 2 \overset{\circ} P_d f =0
\]
by (\ref{eq21a}). We notice that $\psi^1$ is not uniquely determined. In fact, if $h \in {\cal M}^4$ , the function $\psi^1 + h$ solves the original equation (\ref{eq47}), too.
Put $r_i=s_i - \tfrac{4}{R} \left\langle  \lambda^i  \psi^1 \right\rangle^{\circ} $ in relation
(\ref{eq32b}) and apply \remref{REMinhom}. Using the corresponding $h$ from (\ref{eq31a}),
we obtain a solution $\psi^1 + h$ of the linear system ${\cal A}^i[\psi]=s_i$, $i=1,2,3$.
\end{proof}

\section{The nonlinear problem}

\subsection{Application of the implicit function theorem}

The properties of the linearized operator ${\cal A}$ listed in the fundamental
\thmref{Thm2} enable us to apply a variant of the implicit
function theorem, namely the {\em surjective implicit
function theorem} (cf. \cite[Vol. 1, pp 176)]{Z}. We consider the
nonlinear operator ${\cal F}$ in (\ref{eq14}) as a mapping which acts
on the coordinates $x^{\alpha} \in H^7$, on the coordinates
$g_{\alpha \beta} \in H^6$ of the metric tensor $g,$ and on the
parameters $T$, $R.$

\begin{thm}
\label{Thm3} There is a constant $\varepsilon >0$ such that, for
each metric $g$ of $M$ with $\Vert g - \eta \Vert
_{H^6} < \varepsilon$ and each $T \in \mathbb{R},R>0$, rigid spheres ${\cal S}$ in
the sense of \defnref{rigid} exist. These spheres are pertubations
of the round sphere $S_{T,R}$ and depend on additional eight
parameters.
\end{thm}

\begin{proof}
  {\em a)} We show that the nonlinear equation (\ref{eq14}) possesses a solution of
the form $x = \xi + \psi^{\alpha} {\boldsymbol \beta}_{\alpha}
(\xi)$ due to
(\ref{eq25}) and apply Theorem 4.H from \cite[Vol.~1, pp 176]{Z}.
To check  the assumptions of this theorem, let the operator ${\cal
F}$ be defined on the Hilbert space
$(H^7)^4 \times (H^6)^{10} \times \mathbb{R}^2$.
Then ${\cal F}$ maps this space into the closed subspace
\[
 {\cal H} := T^6 \times {\mathbb R}^3 \times
(I- \overset{\circ}P_{md}) H^5 \times (I- \overset{\circ}P_{md}) H^3
\]
of $(H^6)^2 \times {\mathbb R}^3 \times H^5 \times H^3$. This follows
from the definition of ${\cal F}$ in (\ref{eq14}).

  {\em b)} The crucial point is the surjectivity of
${\cal A}: \ D({\cal A}) \rightarrow {\cal H}$ by \thmref{Thm2}c)
and the splitting property of $N({\cal A}) $, i.e.
\[
D({\cal A}) = N({\cal A}) \oplus N({\cal A})^{\bot} \
.
\]
According to this decomposition, we have
\begin{equation}
\label{eq35}
\psi = \rho + \chi \, , \quad \text{where} \quad \rho = (\rho^{\alpha}) \in N({\cal A}) \, ,
\ \chi = (\chi^{\alpha}) \in N({\cal A})^{\bot} \, .
\end{equation}
The cited Theorem 4.H yields the existence of a solution $\chi$ of
\begin{equation}
\label{eq36} {\cal F} \left( \xi + \rho^{\alpha} {\boldsymbol \beta}_{\alpha}(\xi) + \chi^{\alpha} {\boldsymbol \beta}_{\alpha}(\xi) ,g,T,R \right) = 0 \ .
\end{equation}
The solution depends on $\rho, g,T,R$ and is defined
on a neighbourhood of $(0,\eta,T,R)$. Moreover, at this point, the vector field
$\chi$ vanishes.
\end{proof}

\begin{rem}
\label{Rem6}
The idea of application of the surjective implicit
function theorem to prove our main theorem is equivalent to the considerations in \cite[Subsection 3.3]{GJKL}, where H{\"o}lder spaces are substituted with appropriate Sobolev spaces. We point out that an analogous result to Theorem \ref{Thm3} does also hold within
H{\"o}lder framework.
\end{rem}

\subsection{Approximation of the solution}

For the readers convenience let us sketch how to construct a solution to the nonlinear operator equation (\ref{eq14}). This procedure also yields a method to approximate this solution successively. It is the so-called \emph{Newton's
simplified method} for (\ref{eq36}), which produces solutions of the form (\ref{eq35}).

Let a metric $g$ of $M$ with $\Vert g - \eta \Vert
_{H^6} < \varepsilon$ be given, fix $T \in \mathbb{R},R>0$ and $\rho   \in N({\cal A})$. Denote the approximation of the solution $\chi$ of (\ref{eq36}) in the $n$-th iteration step by $\chi^{(n)}$. Then,
it is recursively defined by
\begin{equation}
\label{appr}
\chi^{(n+1)} := \chi^{(n)} - {\cal B}^{-1} {\cal F} \left( \xi + \rho + \chi^{(n)} ,g,T,R \right) \quad \text{for} \ n =0,1,\ldots \ ; \ \chi^{(0)}=0 \, ,
\end{equation}
where ${\cal B}$ is the restriction of the linearized operator ${\cal A}$ to the orthogonal complement $N({\cal A})^{\bot}$ in $D({\cal A})$. Since $\chi^{(n)} \in N({\cal A})^{\bot}$, relation (\ref{appr}) is equivalent to
\[
 {\cal A} \left[  \chi^{(n+1)} - \chi^{(n)} \right] = {\cal B} \left[ \chi^{(n+1)} - \chi^{(n)} \right]
= {\cal F} \left( \xi + \rho + \chi^{(n)} ,g,T,R \right) \ .
\]
Hence, we determine a solution $\delta^{(n)}$ of the inhomogeneous problem $ {\cal A}[\delta^{(n)}] = {\cal F}^{(n)} $ from the Hilbert space $(H^7)^4 $, i.e., in each iteration step we solve the same linear equations on $S^2$ with different right-hand sides
\[
{\cal F}^{(n)} = \left(\t0 s_{AB}^{(n)}, {\bf X}^{(n)},
(I- \overset{\circ}P_{md})k^{(n)} , (I- \overset{\circ}P_{md}) q^{(n)} \right) :=   {\cal F} \left( \xi + \rho + \chi^{(n)} ,g,T,R \right) \ .
\]
This can be carried out by the calculations described in Subsection 4.2 and in the proof of \thmref{Thm2}c).
Explicitly, we get
\begin{align*}
\text{first } {\delta^{(n)}}^0 & \text{ from} \quad {\cal A}_{tors}[\delta^{(n)}] = (I- \overset{\circ}P_{md}) q^{(n)} \ , \\
\text{then } {\delta^{(n)}}^A & \text{ from} \quad {\cal A}_{AB}[\delta^{(n)}] = \t0 s_{AB}^{(n)} \ , \ A,B=2,3 \ , \\
\text{finally } {\delta^{(n)}}^1 & \text{ from} \quad {\cal A}_{curv}[\delta^{(n)}] =
(I- \overset{\circ}P_{md})k^{(n)} \ , \\
& \ \text{ and} \quad {\cal A}^i[\delta^{(n)}] = {{\bf X}^{(n)}}^i \ , \ i=1,2,3 \ .
\end{align*}
Note that, due to \thmref{Thm2}b), the general solution $\delta^{(n)}$ contains eleven free parameters. With $\delta^{(n)}$ we obtain the next correction of approximation $\chi^{(n)}$ defined in (\ref{appr}) by
\[
 \chi^{(n+1)} = \chi^{(n)} + (I- P_{N({\cal A})}) \delta^{(n)} \, ,
\]
where $P_{N({\cal A})}$ denotes the projection to the null space $N({\cal A})$ in the domain $D({\cal A})$ of the linearized operator ${\cal A}$. This projection rules out the ambiguity in the  parameters and $\chi^{(n+1)}$ is uniquely determined.

\begin{rem}
\label{REMKonv} The convergence of the  approximating sequence  $\chi^{(n)}$ to a solution $\chi$ of (\ref{eq36})
follows from Problem 5.1 in \cite[Vol.~1, p.~214]{Z} and is based on the Banach fixed point theorem.
\end{rem}

\begin{rem}
\label{REMeleven} In the eleven-dimensional space of solutions, there is a three-dimensional subspace corresponding to pure rotations. Hence, the space of different rigid spheres is eight-dimensional.
\end{rem}

\paragraph*{Acknowledgement.}
We are deeply indebted to Eberhard Zeidler, Max-Planck-Institute for Mathematics in the Sciences, Leipzig, for fruitful discussions and valuable help. This work was supported in part by Narodowe Centrum Nauki (Poland) under Grant No.
DEC-2011/03/B/ST1/02625.




\end{document}